\documentclass[aps,prl,twocolumn,epsf,epsfig,amsmath]{revtex4}
\usepackage{latexsym}
\usepackage{hyperref}
\usepackage{times}
\usepackage{graphicx}
\usepackage{subcaption}
\usepackage{amsmath}
\usepackage{multirow} 
\usepackage{amsthm}
\usepackage{dcolumn}% Align table columns on decimal point
\usepackage{bm}% bold math
\usepackage{ textcomp }
\usepackage{color}
\usepackage{amssymb}
\usepackage{xcolor}
\usepackage{easyReview}
\usepackage{mathdots}
\usepackage{float}
\usepackage{lineno}
\usepackage{todonotes}
\usepackage{array}
\usepackage{braket}
\usepackage{ragged2e}
\usepackage[justification=RaggedRight,singlelinecheck=false]{caption}

\newtheorem{theorem}{Theorem}

\newtheorem{lemma}[theorem]{Lemma}

\begin{document}
\title{Quantum Metric Bound State of Light}
\author{Jinchao Zhao}
\email{jinchao@ust.hk}
\author{Rongning Liu, Xue-Yang Song}
\author{K.T. Law}
\email{phlaw@ust.hk}
\affiliation{Department of Physics, Hong Kong University of Science and Technology, Clear Water Bay, Hong Kong, China}
\begin{abstract}
    The spatial confinement of defect-induced bound states is conventionally governed by the effective mass in dispersive bands. More recently, Compact Localized States (CLSs) arising from exact destructive interference have been utilized to achieve confinement in flat bands. However, CLSs rely on pristine lattice symmetries and fine-tuned defect profiles. The introduction of a generic local impurity inevitably breaks these strict phase-matching conditions, resulting in extensive bound states whose fundamental length scale has remained an open question. Here, we establish a third regime of confinement: the quantum metric bound state. We provide a rigorous mathematical proof demonstrating that in the absence of kinetic energy and CLS protection, the exponential decay length of these states is lower-bounded by the quantum metric of the unperturbed flat band. We demonstrate the tightness of this geometric limit by constructing a family of highly tunable flat-band generators, and we verify its universality across diverse realistic architectures. Ultimately, this classification establishes the independently measurable quantum metric as a predictive design principle for engineering confined modes in synthetic wave platforms.
\end{abstract}
\maketitle

\textit{Introduction}—
The engineering of flat-band systems stands as a transformative frontier in photonics \cite{xia2018unconventional,myoung2019flat,tang2020photonic,vicencio2021photonic,xia2023photonic,danieli2024flat}, offering a powerful platform to manipulate light-matter interactions through the complete quenching of kinetic energy. Unlike traditional dispersive bands, where transport is governed by group velocity, flat bands possess a vanishing dispersion across the entire Brillouin zone, theoretically resulting in an infinite effective mass. This extreme limit has been successfully realized across a diverse array of photonic settings, ranging from femtosecond-laser-written waveguide arrays \cite{vicencio2015observation,mukherjee2015observation} to exciton-polariton lattices \cite{baboux2016bosonic,whittaker2018exciton} and moiré superlattices \cite{wang2020localization,dong2021flat}, facilitating the exploration of slow-light phenomena, topological protection, and strongly correlated bosonic phases.

Historically, the understanding of wave localization in these architectures has been resolved into two primary regimes. In conventional dispersive bands, localization is purely energetic; the spatial extent of a defect-induced bound state is governed by the effective mass, yielding a decay length $\xi_b\sim1/\sqrt{m^*E_b/\hbar^2}$ (Fig.~\ref{fig:concept}a). Conversely, in flat bands, the divergence of the effective mass gives rise to a celebrated hallmark of flat-band physics: Compact Localized States (CLSs) \cite{maimaiti2017compact,leykam2018perspective,rhim2019classification}. These are macroscopically localized eigenmodes strictly confined to a finite number of unit cells due to exact destructive interference (Fig.~\ref{fig:concept}b). As an ideal platform for robust spatial confinement, CLSs have been widely studied for applications ranging from diffraction-free wave packet propagation\cite{xia2025fully} and high-efficiency image transmission \cite{xia2016demonstration} to low-threshold micro-lasing \cite{longhi2019photonic} and the enhancement of nonlinear optical effects \cite{tang2020photonic}. 

Crucially, however, CLSs describe highly specific, fine-tuned localized profiles or symmetric multi-site clusters engineered to maintain perfect phase cancellation \cite{vicencio2015observation,mukherjee2015observation}. When a generic, single-site local impurity is introduced, these phase-matching conditions are naturally broken \cite{danieli2024flat}, giving rise to extensive bound states that escape the CLS paradigm (Fig.~\ref{fig:concept}c). Defining the fundamental length scale that governs the bound state localization when kinetic energy is quenched, but CLS protection is absent, remains a critical open question. In this work, we show that the hidden length scale defined by quantum metric \cite{provost1980riemannian,anandan1990geometry,marzari2012maximally,rossi2021quantum,yu2025quantum,liu2025quantumold,verma2026quantum}, which termed the quantum metric length $l_{qm}$ \cite{chen2024ginzburg,hu2025anomalous,li2025flat,guo2025majorana,onishi2025quantum,lee2025embedding,kim2025real,chen2026quantum,chau2026quantum,dai2026quantum}, determines the exponential decay length of a generic bound state such that $\xi_b \ge 4/N^{2} l_{qm}$ (Fig.~\ref{fig:concept}c). When the quantum metric length is the dominating length scale governing the decay lengths of the bound states, we term these bound states the \textit{quantum metric bound states}.

\begin{figure*}[htbp]
    \centering
    \includegraphics[width=17.2cm]{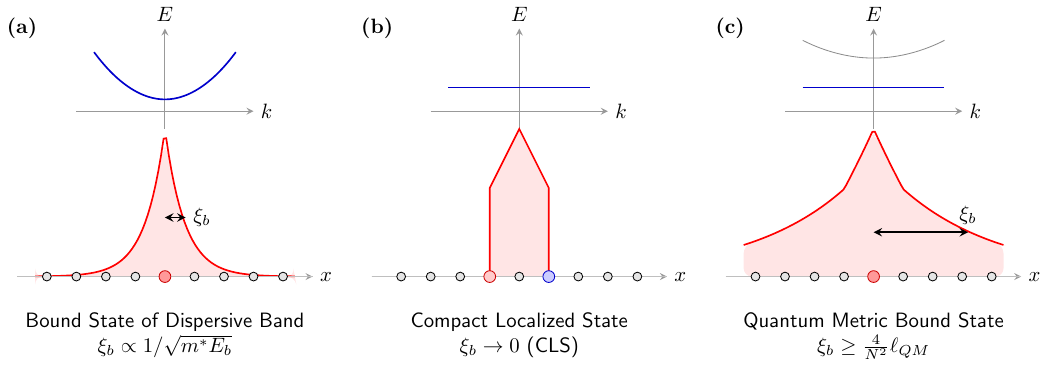}
    \caption{(a) In traditional dispersive bands, the spatial extent ($\xi_b$) of a defect-induced bound state is governed by the effective mass $m^*$ and the bound state energy $E_b$. (b) In an exact flat-band, the Compact Localized States are achieved given that the defect profiles are fine-tuned. (c) A generic local impurity breaks the phase-matching conditions for a CLS, and the bound state develops an extensive evanescent tail, bounded from below by the quantum metric length $\ell_{QM}$ of the unperturbed flat band. Here $N=(n-1)m$, where $n$ is the number of bands, and $m$ is the largest hopping range of electrons.}
    \label{fig:concept}
\end{figure*}

The hidden length scale, defined by the band wave functions and called the quantum metric length, was first introduced in the seminal work by Marzari et al. in Ref.\cite{marzari1997maximally}. It was shown that the quantum metric length governs the quadratic spread of the maximally localized Wannier functions\cite{marzari2012maximally}. While its broader physical consequences initially remained unexplored, recent developments have demonstrated that in flat band materials with vanishing Fermi velocity, the quantum metric length becomes the dominating length scale in many solid-state systems. For example, the quantum metric length determines the superconducting coherence length in flat band superconductors\cite{ chen2024ginzburg, hu2025anomalous, chen2026quantum}. It also governs the coherence length of flat band Josephson junctions\cite{li2025flat}, the spatial decay length of Majorana zero modes in topological superconductors\cite{guo2025majorana}, and the diffusion length as well as the localization length in disordered flat band systems\cite{chau2026quantum,dai2026quantum}. 

Despite all of the above works, a rigorous and model-independent theory linking the quantum metric tensor to the impurity-induced bound state remains elusive. In this Letter, we resolve this conceptual gap by establishing that the quantum metric as a universal geometric bound on wave localization. Specifically, we present a rigorous, model-independent proof that the exponential decay length ($\xi_b$) of a bound state induced by an arbitrary local impurity is fundamentally lower-bounded by the integrated quantum metric ($\ell_{QM}$) of the unperturbed flat band. By connecting this spatial decay directly to the analytic pole structure of the band projection operators, we utilize matrix Bernstein-type inequalities to establish the geometric tensor as a physical floor on confinement. To demonstrate the impact of this geometric limit beyond electronic solid-state systems, we proposed a realistic optical lattice architecture that allows for the precise tuning and direct observation of metric-controlled light confinement. Because this quantum metric bound state is generally an intrinsic wave phenomenon, our findings establish the quantum metric length as a predictive design principle for engineering confined modes in next-generation synthetic wave platforms.

\textit{Photonic Flat Band}—
To physically realize the flat-band dynamics established above, we consider the propagation of light within an array of evanescently coupled optical waveguides. Assuming a fixed polarization state and weak-guidance conditions, the electric field can be treated within the scalar paraxial approximation and is expressed as $E(x,y,z)=\Psi(x,y,z)e^{ik_0 n_0z}$, where $k_0$ is the vacuum wavenumber and $n_0$ is the background refractive index. Under the paraxial approximation, the slowly varying envelope $\Psi(x,y,z)$ follows the Schrödinger-like equation:
\begin{equation}\label{eq:phsch}
    i\partial_z\Psi(x,y,z)=\Big[-\frac{1}{2 k_0 n_0}\nabla^2_\bot -\frac{k_0}{n_0}\Delta n(x,y,z)\Big]\Psi(x,y,z),
\end{equation}
where the refractive index profile $\Delta n(x,y,z)=n(x,y,z)-n_0$ acts as an effective optical potential, and the propagation coordinate $z$ serves as an effective time axis. Looking for stationary solutions, we decompose the envelope as $\Psi(x,y,z) = \psi(x,y) e^{i\beta z}$, with $\beta$ being the propagation constant. Because $\Delta n(x,y,z)$ is uniform along $z$ and periodic in the transverse plane, Bloch’s theorem applies. The eigenmodes take the form $\psi_{n,\boldsymbol{k}}(x,y)=e^{i\boldsymbol{k}\cdot \boldsymbol{r}}u_{n,\boldsymbol{k}}(x,y)$ with $u_{n,\boldsymbol{k}}(x,y)$ having the same periodicity as the lattice. This yields a highly controllable band structure $\beta_n(\boldsymbol{k})$.

In experimental settings, these waveguides are evanescently coupled, allowing the continuous photonic lattice to be accurately captured by a discrete tight-binding model. The specific architecture considered here is a one-dimensional variant of the Lieb lattice \cite{chau2026quantum,mukherjee2015observation}. As illustrated in Fig.~\ref{fig:model}a, each unit cell consists of three distinct waveguides—labeled A, B, and C—with corresponding annihilation operators $a, b,$ and $c$. The Hamiltonian reads:
\begin{equation}
    H(x)=\sum_x J_+(a_x^\dagger b_x+a_x^\dagger c_x)+J_-(a_x^\dagger b_{x+1}+a_x^\dagger c_{x-1})+h.c.,
\end{equation}
where $J_\pm=J(1\pm \delta)$ denote the intra- and inter-cell hopping amplitudes. In our continuous model simulations, this hopping asymmetry $\delta$ is precisely tuned via the geometric angle $\angle \mathrm{B}_{x}\mathrm{A}_{x-1}\mathrm{A}_{x} \equiv \theta$ (Fig.~\ref{fig:model}a), while fixing $\angle\mathrm{B}_{x+1}\mathrm{A}_x\mathrm{C}_x$ at $90^\circ$. Consequently, the intra-cell (blue link) and inter-cell (green link) hopping distances are $a\sin\theta$ and $a\cos\theta$, respectively.

Adapting realistic experimental parameters \cite{mukherjee2015observation}, we set the lattice constant $a = 84\mu\mathrm{m}$, the waveguide diameter $d = 7.4\mu\mathrm{m}$, and the background index $n_0 = 1.5$. The waveguides feature a weak index contrast of $n = 1.5005$. We utilize the COMSOL Wave Optics Module to achieve the exact band structure (Fig.~\ref{fig:model}b). The spectrum exhibits a single, isolated nearly flat band (red) sandwiched between two dispersive bands (black). Crucially, the geometric tuning of $\theta$ directly controls the asymmetry $\delta$, which in turn opens a band gap of order $\sim \delta$. While residual next-nearest-neighbor couplings (between $\mathrm{A}_x\mathrm{A}_{x+1}$ and $\mathrm{C}_x\mathrm{B}_{x+1}$) introduce a slight dispersion, the flat band remains highly isolated, providing an ideal, high-quality platform to probe the geometric limits of light confinement. The quantum metric length $\ell_{QM}$ serves as the fundamental length scale in the neighbour of a flat band\cite{chen2024ginzburg,guo2025majorana,li2025flat,chau2026quantum}. It can be calculated directly from the continuous-wave solutions, avoiding explicit momentum derivatives. By extracting the real-space, cell-periodic Bloch functions $u_{f,k}(\boldsymbol{r})$ of the central band using the COMSOL eigensolver, the quantum metric length is defined via the distance between adjacent momentum states
\begin{equation}
    \ell_{QM} = \int_{-\pi/a}^{\pi/a} \frac{dk}{2\pi} \lim_{\Delta k \to 0} \frac{1 - |\langle u_k | u_{k+\Delta k} \rangle|^2}{\Delta k^2},
    \label{eq:qml}
\end{equation}
where the inner product $\langle u_k | u_{k+\Delta k} \rangle = \int_{\mathrm{U.C.}} u_k^*(\boldsymbol{r}) u_{k+\Delta k}(\boldsymbol{r}) d^2\boldsymbol{r}$ is integrated over a single unit cell.

\begin{figure}[h]
    \centering
        \includegraphics[width=\linewidth]{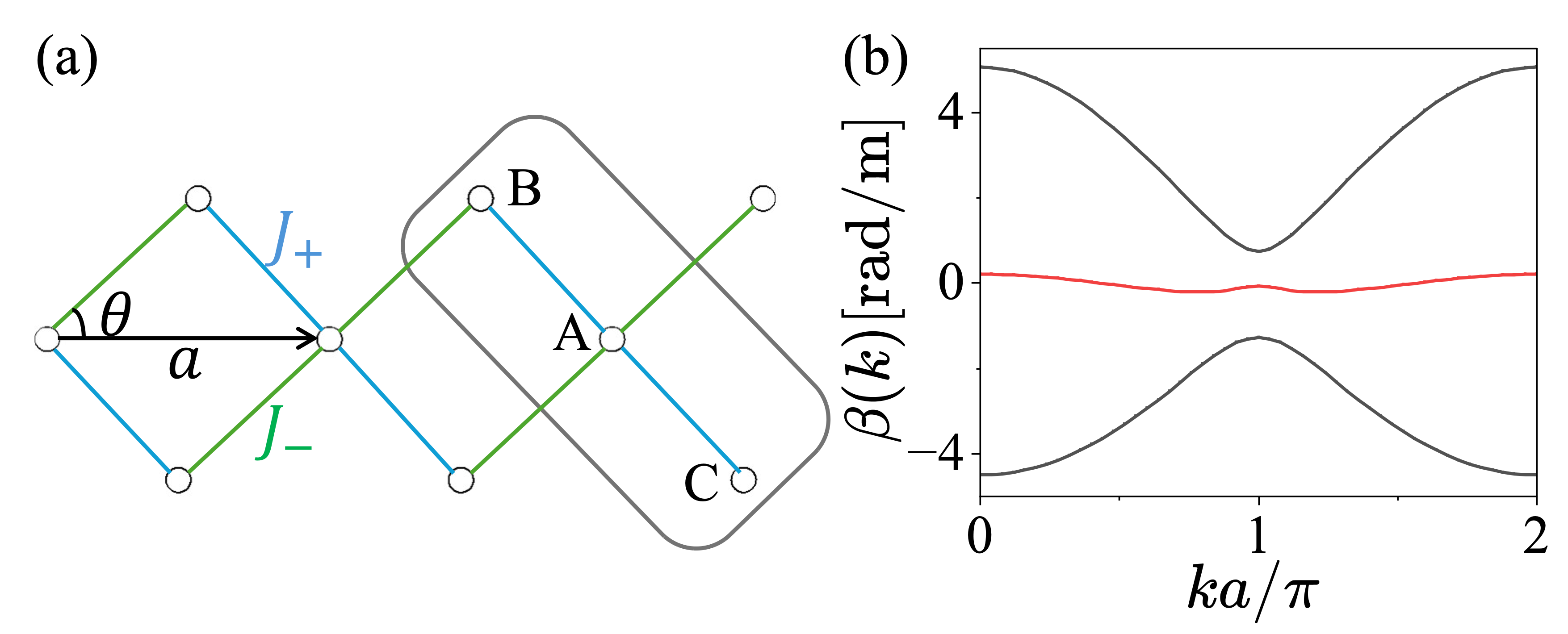}
    \caption{(a) Schematic of the one-dimensional photonic Lieb lattice. Each unit cell (gray box) consists of three waveguides, labeled A, B, and C. The lattice constant is $a = 84\mu\mathrm{m}$, and the waveguide diameter is $d = 7.4\mu\mathrm{m}$. The substrate refraction index is $n_0=1.5$, and the refraction index of waveguides is $n=1.5005$. The angle $\angle\mathrm{B}_{x+1}\mathrm{A}_x\mathrm{C}_x$ is fixed at $90^\circ$. The asymmetry between the coupling strengths $J_+$ and $J_-$ is controlled by the angle $\theta$.  
    (b) The band structure of the one-dimensional photonic Lieb lattice, calculated by COMSOL. We normalize the propagation constant $\beta(k)$ by subtracting the average value of the center band.}
    \label{fig:model}
\end{figure}

\textit{Bound States}—
Having established a controllable flat-band architecture, we now introduce a local symmetry-breaking perturbation to directly probe this third regime of localization. Crucially, we do not engineer a multi-site defect designed to support a CLS; instead, we introduce a generic, single-site local impurity to break the exact destructive interference conditions. By breaking translational invariance, this localized defect scatters the flat-band modes, pulling a bound state into the band gap. We investigate the asymptotic exponential profile of the resulting bound state via the ansatz
\begin{equation}
\psi_b(x-x_0)\sim p(x-x_0)e^{-|x-x_0|/\xi_b},
\label{eq:bound_state}
\end{equation}
where $x_0$ denotes the defect site, $p(x-x_0)$ is a polynomial prefactor dominating short-range oscillations, and $\xi_b$ is the characteristic long-range decay length.

In our continuous-wave simulations, this bound state is generated by slightly reducing the refractive index of a single B-site waveguide to $n_b$ (Fig.~\ref{fig:bound_state}a). To ensure the bound state minimizes hybridization with distant dispersive bands, the defect perturbation is kept perturbatively weak. The index contrast between the defect and normal waveguides is chosen to be significantly smaller than that between the waveguides and the background, such that $\Delta_b\equiv|n_b - n|/\Delta n \ll 1$.

By systematically varying the lattice geometry $\theta$, we extract the decay length $\xi_b$ from the asymptotic tails of the simulated electric field profiles (Fig.~\ref{fig:bound_state}b). Remarkably, despite the quenching of kinetic energy—which classically predicts an infinitely localized state ($\xi_b \to 0$) due to a diverging effective mass—the extracted decay lengths always remain finite. To understand this behavior, we compare the physical decay length $\xi_b$ against the quantum metric length $\ell_{QM}$ of the unperturbed central band. As $\theta$ approaches $45^\circ$, the extracted decay length $\xi_b$ increases monotonically, consistently tracking and staying above a geometric floor defined by the bulk quantum geometry (Fig.~\ref{fig:bound_state}c). This robust numerical correspondence motivates a fundamental question: why does the unperturbed band geometry control the minimal spatial extent of a defect state when kinetic energy is absent? To resolve this question, we must construct a model-independent theoretical framework that maps the analytic structure of the Bloch manifold to real-space localization.

\begin{figure}[h]
    \includegraphics[width=\linewidth]{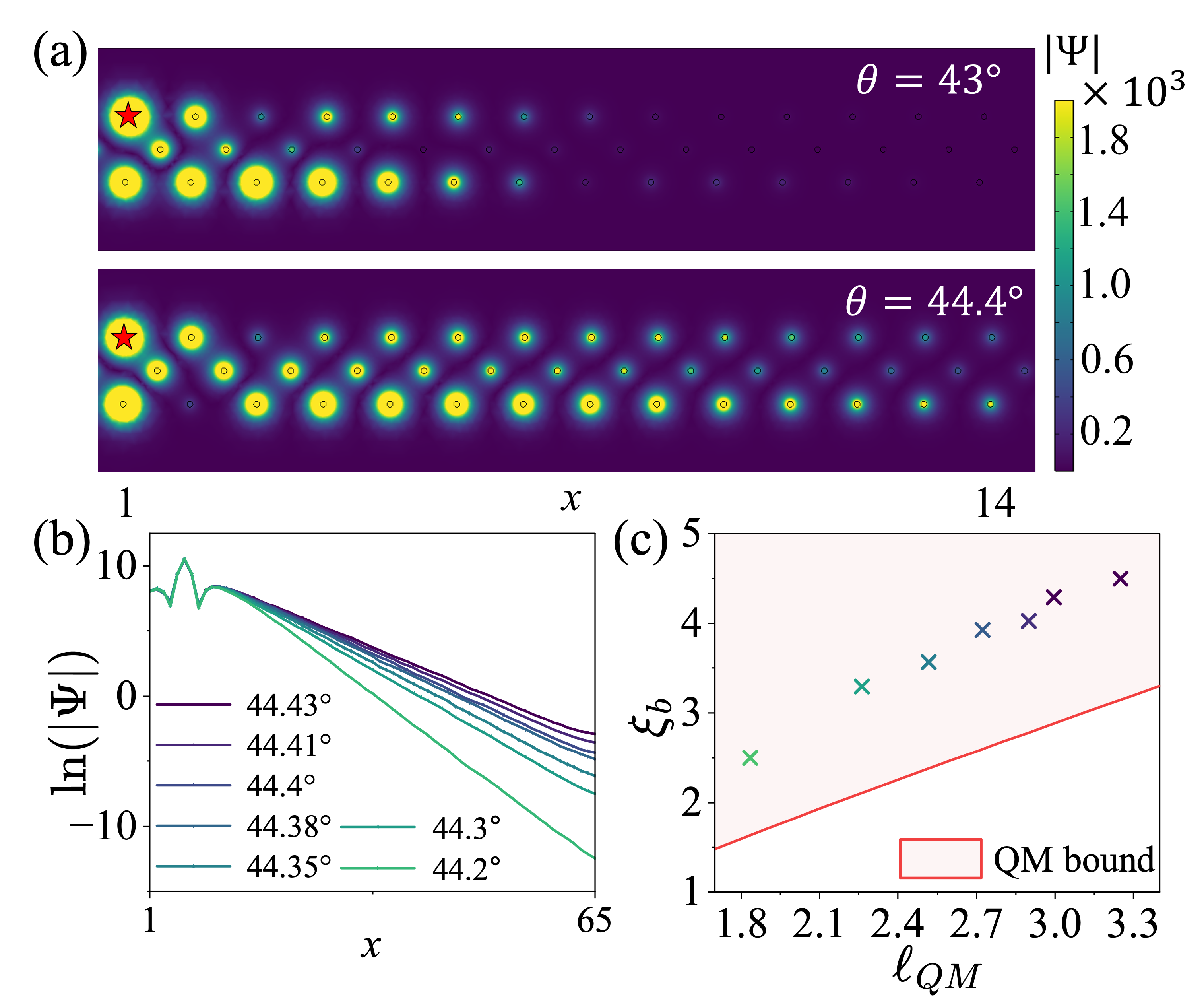}
    \caption{(a) Electric field intensity showing the gradual decay along the unit cells.
    The red star marks the waveguide with the altered refractive index.
    (b) The bound state configurations at the centers of the B sites
    of different values of $\theta$. The short-range oscillations arise from the polynomial prefactor as shown in Eq.~\ref{eq:bound_state}.
    (c) Comparison between $\xi_b$ and $\ell_{QM}$ for different $\theta$ using the same legend color in panel b. The shaded area is defined by Eq.~\ref{eq:qmineq} with $N=2$. All the simulations are performed under the weak impurity limit by setting $\Delta n_b / \Delta n = 0.05\%$,  
    }
    \label{fig:bound_state}
\end{figure}

\textit{Theoretical Framework and Geometric Bounds}—
To establish the universal mechanism governing this quantum-metric-controlled confinement, we examine the analytic structure of the defect-induced bound states. In a generic multi-band lattice, introducing a local impurity scatters the flat-band modes to pull a discrete bound state into the gap. Following the Lifshitz-Koster-Slater formalism, the real-space profile of this state is determined by the unperturbed bulk Green's function(See End Matter)
\begin{equation}
    G^f(R;E)= \frac{1}{V_k}\int dk e^{ik\cdot R}\frac{P_f(k)}{E-E_{f}}
\end{equation}
where $P_f(k) = \ket{\psi_{f,k}}\bra{\psi_{f,k}}$ is the flat-band projection operator. Because the flat-band energy $E_f$ is independent of momentum, the kinetic energy denominator factors out of the integral completely. Consequently, the spatial structure of $G^f(R;E)$ is independent of traditional effective mass or dispersion. It depends entirely on the Hilbert space geometry information encoded in $P_f(k)$.

Analytically continuing the momentum into the complex plane ($k = x + iy$), the real-space spatial decay for a unit cell separation $R>0$ is evaluated via contour integration. The asymptotic decay length is governed by the residue of the complex pole closest to the real axis, scaling as $\xi_b \sim 1/y_{\mathrm{min}}$, where $y_{\mathrm{min}}$ is the distance of the nearest singularity from the real momentum axis. The physical limit of light confinement thus maps onto an algebraic question: what bounds the proximity of these complex poles to the real axis?

We resolve this by linking the pole locations to the momentum-space variations of the Bloch states. The bulk quantum metric, defined as $\mathcal{G}(k)=\frac{1}{2}\mathrm{Tr}{[(\partial_k P_f)^2]}$, measures the gauge-invariant distance between adjacent eigenstates in Hilbert space. For a finite-range tight-binding system with maximum hopping range $m$ and $n$ bands, the elements of $P_f(k)$ evaluate to rational trigonometric functions of maximum algebraic degree $N=(n-1)m$ (Lemma~\ref{lemma:projection_poles} in End Matter). According to Bernstein-type approximation theory, a rational function matrix with a real-space cutoff cannot exhibit arbitrary rapid phase winding along the real axis; its derivative is capped by the location of its complex poles\cite{borwein2012polynomials}. By integrating this derivative bound across the Brillouin zone, we find that the integrated quantum metric length $\ell_{QM}=\int \frac{dk}{2\pi}\mathcal{G}(k)$ is rigorously constrained by $y_{\mathrm{min}}$ through 
\begin{equation}
    \ell_{QM}\le\frac{N^2}{4}\coth{y_{\mathrm{min}}},
\end{equation}
See Theorem~\ref{theorem:qml_bound} in End Matter.
Crucially, in the trivial limits of an atomic chain ($m=0$) or a single isolated band ($n=1$), where the algebraic degree $N=(n-1)m$ vanishes, the inequality reduces consistently to $0 \le 0$, since in both cases the quantum geometry also naturally vanishes ($\ell_{QM}=0$). 
In the weak impurity limit where the bound state is energetically close to the flat band, $\xi_b \sim 1/y_{\mathrm{min}}$. Expanding this inequality yields a universal geometric lower bound on wave confinement
\begin{equation}
    \xi_b \gtrsim \left(\mathrm{Arccoth}\left(\frac{4}{N^2}\ell_{QM}\right)\right)^{-1} \stackrel{\ell_{QM}\rightarrow\infty}{\Longrightarrow} \frac{4}{N^2}\ell_{QM}.
\label{eq:qmineq}
\end{equation}
This inequality proves that the spatial extension of a flat-band mode cannot be suppressed by arbitrarily deep defect potentials. Instead, the absolute minimum spread of the evanescent tail is determined by the unperturbed bulk quantum metric length $\ell_{QM}$. For the photonic Lieb lattice evaluated in the previous example ($n=3, m=1$), the degree evaluates to $N=2$. As plotted in Fig.~\ref{fig:bound_state}c, the geometric floor is respected by our continuous-wave simulations. Rather than acting as an abstract theoretical descriptor, the quantum metric functions as a predictive, independently measurable design limit for wave localization in synthetic structures.

To demonstrate that this quantum geometric limit is strict, meaning the inequality is tight and physically saturable, we construct an exact tight-binding model. By engineering the flat-band projection operator from a degree-$N$ Blaschke product, we force its complex poles to lie precisely at the extreme boundary permitted by the Bernstein-type derivative bound\cite{borwein2012polynomials}. For the minimal case of $N=1$, we derive a two-band ($n=2$), nearest-neighbor ($m=1$) tight-binding Hamiltonian(see End Matter for details):
\begin{equation}
H_{1}(k)=C\left(\begin{array}{cc}
\cos k-\cosh y_{0}&1-\cos(k-iy_{0})\\
1-\cos(k+iy_{0})&\cos k-\cosh y_{0}
\end{array}\right),
\end{equation}
where $C$ defines the energy scale and $y_0$ parameterizes the isolating band gap. By construction, the distance of these complex poles to the real axis perfectly saturates the derivative bound, causing the physical bound-state decay length to exactly equal the geometric quantum metric limit.

Beyond serving as a rigid mathematical validation, this approach introduces a powerful paradigm for flat-band lattice generation. Traditional flat-band design strategies predominantly rely on real-space CLSs\cite{maimaiti2017compact,maimaiti2021flat,chen2023decoding,liu2025designing}, a methodology that inherently obscures control over momentum-space quantum geometry. Conversely, brute-force spectral flattening via algebraic division inevitably introduces unphysical, infinite-range hoppings \cite{neupert2011fractional,parameswaran2013fractional}. In contrast, our Blaschke-product generator yields tight-binding models that are short-ranged, analytically tractable, and highly tunable via the parameter $y_0$. These Hamiltonians are geometrically optimal, maximizing the bulk quantum metric mathematically permissible for a given hopping range and bound-state decay profile.

The physical significance of this geometric bound state is manifested by its direct experimental accessibility in state-of-the-art photonic platforms, such as femtosecond-laser-written waveguide arrays \cite{vicencio2015observation,mukherjee2015observation}. Within these architectures, a generic local impurity can be easily engineered by modulating the laser writing speed or power for a single waveguide, precisely tuning its local on-site potential (refractive index). The exponential decay length $\xi_b$ of the resulting evanescent mode can be directly mapped via near-field scanning optical microscopy or direct output facet imaging. Concurrently, the unperturbed quantum metric $\ell_{QM}$ can be independently extracted in the pristine lattice using established momentum-space interferometry \cite{asteria2019measuring} or by tracking the anomalous transverse drift of wave packets \cite{gianfrate2020measurement}. Because both sides of the inequality $\xi_b \ge \frac{4}{N^2}\ell_{QM}$ are independently measurable, the bound serves as a concrete, testable law linking abstract Hilbert-space topology to real-space spatial dynamics.

In realistic experimental settings, weak residual long-range couplings inevitably introduce a small but finite bandwidth $W_f$, rendering the flat band quasi-flat. Although the algebraic derivation of the projector degree $N=(n-1)m$ assumes an exact dispersionless band, the physical relevance of the bound remains robust through the continuity of the resolvent operator. A realistic system can be expressed analytically as $H_{\text{exp}}(k) = H_{\text{ideal}}(k) + \delta H(k)$, where $\delta H(k)$ accounts for the finite-range residual couplings. The spatial decay of the experimental bound state is governed by the complex roots of the characteristic polynomial $\det(E_b I - H_{\text{ideal}}(k) - \delta H(k)) = 0$. Because the roots of a polynomial depend continuously on its coefficients, a weak perturbation ($W_f \ll \Delta$, where $\Delta$ is the isolating band gap) only smoothly perturbs the complex singularities of the unperturbed system. Consequently, the physical decay length $\xi_{\text{exp}}$ acquires only a finite, perturbative correction of order $\mathcal{O}(W_f/\Delta)$, ensuring that the inequality $\xi \gtrsim \frac{4}{N^2}\ell_{QM}$ acts as an asymptotic floor even in non-ideal experimental environments.

\textit{Discussion}—
The establishment of this geometric limit provides a comprehensive taxonomy for wave localization in flat-band architectures, resolving it into three distinct regimes. While traditional dispersive bands respect an energetic localization scaling as $\xi \propto 1/\sqrt{m^* E_b}$ (Fig.~\ref{fig:concept}a), and highly symmetric lattices support truncated Compact Localized States (CLSs) via fine-tuned destructive interference (Fig.~\ref{fig:concept}b), a generic local impurity inevitably breaks this phase cancellation \cite{danieli2024flat}. In this third regime—where kinetic energy is quenched, but CLS protection is lost—the spatial extent of the evanescent tail is governed by the analytic pole structure of the band projection operator, enforcing a geometric floor $\xi \gtrsim \frac{4}{N^2}\ell_{QM}$ (Fig.~\ref{fig:concept}c). Because our proof relies on the algebraic properties of the projection operators, this geometric bound is universal. It holds seamlessly in higher dimensions: the directional decay length is constrained by the integrated trace of the metric tensor ($\xi_\mu \gtrsim \frac{4}{N^2}\ell_{QM,\mu}$) along the $\mu$ direction. It also applies broadly to any physical system governed by dispersionless wave physics, including classical acoustic \cite{han2025all} and mechanical metamaterials \cite{craster2023mechanical}, as well as quantum plasmonic arrays \cite{xu2022topological} and exciton-polariton condensates \cite{byrnes2014exciton}. This metric-mandated spatial leakiness is not an unwanted flaw, but a robust resource: it guarantees a minimum modal overlap that can be directly harnessed to optimize dipole-dipole interactions between embedded emitters \cite{sun2024ultra}, minimize mode volumes in flat-band micro-lasers \cite{longhi2019photonic}, and enhance the interaction cross-section for optical sensing applications \cite{wang2026plasmonic}.

Crucially, this geometric floor bridges individual local confinement properties to macroscopic boundary phenomena. In a finite-sized flat-band lattice, topological edge modes or boundary states localized at opposite ends of the crystal develop extensive evanescent tails into the bulk that cannot be compressed below the metric-defined limit. This guaranteed spatial penetration enables remote boundary states to hybridize across the bulk over macroscopically long distances, completely bypassing the constraints of a vanishing group velocity. This phenomenon directly parallels recent discoveries in topological quantum platforms, where the quantum metric length governs the long-range hybridization of Majorana zero modes and mediates ultra-long-range non-local crossed Andreev reflections \cite{guo2025majorana}. In synthetic wave platforms, translating this mechanism offers a robust design principle for engineering coherent, non-local optical switches and directional couplers whose cross-talk is protected by the bulk geometry.

Finally, the tight-binding models engineered to saturate this geometric bound provide a powerful framework for future research. Because these models maximize the quantum metric while maintaining short-range hoppings, they are ideal for quantum metric engineering. Their structural simplicity and high tunability facilitate experimental realization across synthetic macroscopic platforms. Furthermore, this optimal, short-range nature makes these Hamiltonians well-suited to serve as the non-interacting foundation for exploring how the quantum metric determines localization \cite{dai2026quantum} and anomalous transport in strongly correlated, many-body interacting regimes. In all, these findings establish the quantum metric as a universal, predictive design principle for wave confinement, bridging the abstract geometric information of Hilbert space with the tangible limits of real-space spatial dynamics.

\textit{Acknowledgements}—
   The authors appreciate helpful discussions with Yixin Xiao, Tianyue Li, and Changhao Meng, and especially thank C.T. Chan for fruitful discussions. J. Z. and K. T. L. acknowledge the support of the Ministry of Science and Technology, China, The New Cornerstone Foundation, and the Hong Kong Research Grants Council through Grants No. MOST23SC01-A, No. RFS2021-6S03, No. C6053-23G, No. AoE/P-701/20, AoE/P-604/25R, No. 16309223, No. 16311424 and No. 16300325. R. L. and X.-Y. S. acknowledge the support of the Early Career Scheme of the Hong Kong Research Grants Council with grant No. 26309524.

\bibliography{new}

\section{end matter}

\textit{Bound state in multi-band systems}—
In the following, we assume the system has multiple bands, and a weak local impurity is introduced. 
The Green function is given by
\begin{equation}
\label{eq:green}
    G_{0}(R;E)= \frac{1}{V_k}\sum_n\int dk e^{ik\cdot R}\frac{P_n(k)}{E-E_{nk}},
\end{equation}
where $P_n(k)=\ket{\psi_{nk}}\bra{\psi_{nk}}$ is the projection operator to the $n$th band.
Following the Lifshitz-Koster-Slater method, for a local impurity $V=\sum_{\alpha\beta}V_{\alpha\beta}\ket{0\alpha}\bra{0\beta}$, the bound state wave function is given by
\begin{equation}
    \psi_b^\alpha(r)\equiv \braket{r\alpha|\psi_b}=\sum_{\beta\gamma}G_{0\alpha\beta}(r;E_b)V_{\beta\gamma}{\psi_b^\gamma(0)}.
\end{equation}
As the linear combination, the wave function follows the same decay behavior as the real-space Green function. 
According to Eq.~\ref{eq:green}, the Green function has contributions from all of its band components,
\begin{equation}
    G_0^n(R,E)=\frac{1}{V_k}\int dk e^{ik\cdot R}\frac{P_n(k)}{E-E_{nk}}.
\end{equation}

For the system with an isolated flat band, the Green function has a contribution 
\begin{equation}
    G_{0}^f(R;E)=\frac{1}{V_k}\frac{1}{E-E_f}\int dk e^{ik\cdot R}P_f(k)
\end{equation}
As $E-E_f$ has no $k$ dependence, the real space decay behavior of $G^f$ is determined by the pole structure of $P_f(k)$ on the complex plane of $k$. Due to the $2\pi$ periodicity $H(k)=H(k+2\pi)$ on the real axis, the real space Green function for $R>0$ reduces to the residual contribution from the poles $\left\{k_{f,\ell}^+\right\}_{\ell=1}^{p_f}\subset\mathbb{C}^+$ in the upper half complex plane of the flat band projection operator,
\begin{equation}
\begin{split}
    G_{0}^f(R>0;E)&=\frac{1}{V_k}\frac{1}{E-E_f}\sum_{\ell=1}^{p_f}\mathrm{Res}_{k=k^+_{f,\ell}}\left[ e^{ikR}P_f(k)\right]\\
    &\propto \sum_{\ell=1}^{p_f} e^{ik_{f,\ell}^+R}\;\sim\; e^{-y_{\mathrm{min}}R},
\end{split}
\end{equation}
where $y_{\mathrm{min}}$ is the imaginary part of the pole $k_{f,\ell}^+=x_{f,\ell}+iy_{f,\ell}$ that is closest to the real axis.

In order to get an estimate on the real space decay of the flat band projection operator, we adopt the Bernstein-Type Inequality:
\begin{lemma}
\label{lemma:Bernstein}
    Given $(a_\ell)_{\ell=1}^{2N}=(x_\ell+iy_\ell)_{\ell=1}^{2N}\subset\mathbb{C}\setminus\mathbb{R}$ and $K=\mathbb{R}(\mathrm{mod}\,2\pi)$, define
    \begin{equation}
        \mathcal{T}_N^c(a_1,\cdots,a_{2N},K)=\left\{t(\theta)/\prod_{\ell=1}^{2N}|\sin((\theta-a_\ell\\)/2)|:t\in\mathcal{T}_N^c\right\},
    \end{equation}
    and
    \begin{equation}
        \mathcal{T}_N^c=\left\{t:t(\theta)=\sum_{\ell=-N}^{N}c_{\ell}e^{i\ell\theta},c_\ell\in\mathbb{C}\right\}.
    \end{equation}
    Define the Bernstein factor
    \begin{equation}
        B_N(\theta)=\max\left\{-\sum_{\substack{
             \ell=1 \\
             y_\ell<0
        }}^{2N}P_{y_\ell}(\theta-x_\ell),\sum_{\substack{
             \ell=1 \\
             y_\ell>0
        }}^{2N}P_{y_\ell}(\theta-x_\ell)\right\},
    \end{equation}
    where $P_y(x)={\sinh{y}}/\left({\cosh{y}-\cos{x}}\right)$ is the Poisson kernel.
    Then 
    \begin{equation}
        |f'(\theta)|\le B_N(\theta)||f||_{K}, \quad \forall\theta\in K,
    \end{equation}
    for every $f\in\mathcal{T}_n^c(a_1,\cdots,a_{2N},K)$.
\end{lemma}
\begin{proof}
    See \cite{borwein2012polynomials} Corollary 7.1.8 (Bernstein-Type Inequality on $K$, Complex Case), pp327.
\end{proof}

\begin{lemma}
\label{lemma:projection_poles}
    For a tight-binding system with $n$ orbitals in each unit cell with maximal hopping range $m$, if this system has one isolated flat band with a finite band gap, the matrix elements of the flat band projection operator $P(k)=|v_0\rangle\langle v_0|$ are rational triagonal functions with order of numerator and denominator $N= (n-1)m$.
\end{lemma}
\begin{proof}
    The Hamiltonian of this system is a $n\times n$ Hermitian matrix $H(k),k\in K:=\mathbb{R}(\mathrm{mod}\  2\pi)$, with matrix elements $H_{ij}(k)\in \mathcal{T}^c_m(k):=\{\sum_{l=-m}^{m}a_le^{ilk},a_l\in \mathbb{R}\}$. The Hamiltonian can be diagonalized 
    \begin{equation}
        H= U D U^{-1},\ D =\mathrm{diag}(\lambda_1=0,\lambda_2,...,\lambda_n).
    \end{equation}
    Here $\lambda_1=0$ as this system has a single flat band. Consider the adjugate matrix $\mathrm{adj}(H)$, we have
    \begin{equation}
        \mathrm{adj}(H)= \mathrm{adj}(U D U^{-1})=U \mathrm{adj}(D) U^{-1}=(\prod_{i\neq 1}\lambda_i)|v_0\rangle\langle v_0|,
    \end{equation}
    Therefore, 
    \begin{equation}
        P=\frac{\mathrm{adj}(H)}{\mathrm{tr}(\mathrm{adj}(H))}.
    \end{equation}
    One the other hand, $\mathrm{adj}(H)_{ij}=(-1)^{i+j}\mathrm{det}(H^{(i|j)})$, we have $\mathrm{adj}(H)_{ij}\in \mathcal{T}^c_{N}, N:=m(n-1)$. Therefore, the numerator and denominator of $P_f$ both belong to $\mathcal{T}^c_{N}$. 
    As the denominator $\prod_{i\neq 1}\lambda_i$ is always real for any real value of $k$, its zeros occur symmetrically above and below the real axis, allowing us to take $y_\ell>0$, $a_\ell=x_\ell+iy_\ell$, $a_{N+\ell}=x_\ell-iy_\ell$ for $\ell=1,2,\cdots, N$, and$P_{i,j}\in\mathcal{T}_n^c(a_1,\cdots,a_{2n},K)$.
\end{proof}

\begin{theorem}
    Given the quantum metric $\mathcal{G}(k)=\frac{1}{2}\mathrm{Tr}{[(\partial_k P_f)^2]}$, the quantum metric length is defined by $\ell_{QM}=\int \frac{dk}{2\pi}\mathcal{G}(k)$. It satisfies 
    \begin{equation}
        \ell_{QM}\le\frac{1}{4}\sum_{i=1}^{N}\sum_{j=1}^{N}P_{y_i+y_j}(x_i-x_j)\le\frac{N^2}{4}\coth{y_{\mathrm{min}}},
    \end{equation}
    where $x_\ell+iy_\ell$ are the poles of the flat band projection operator $P$ in the upper half plane $y_\ell>0$, and $y_{\mathrm{min}}$ is the imaginary part of the pole that is the closest to the real axis.
    \label{theorem:qml_bound}
\end{theorem}
\begin{proof}
    Consider the reflaction operator $R(k)=2P(k)-I$, we define scalar function $$f_{u,v}(k)=u^\dagger R(k)v,$$ where $u$ and $v$ be arbitrary unit vectors in $\mathbb{C}^n$. According to Lemma~\ref{lemma:projection_poles}, every matrix element of the flat band projection operator $P_{ij}(k)\in\mathcal{T}_n^c(a_1,\cdots,a_{2n},K)$, where $a_\ell=x_\ell+iy_\ell$ are the poles defined in Lemma~\ref{lemma:projection_poles}. Thus  $f_{u,v}(k)\in\mathcal{T}_n^c(a_1,\cdots,a_{2n},K)$. 
    As a reflection operator, $R(k)$ has a operator norm $||R(k)||=1$. Thus for any unit vectors $u$ and $v$, $||f_{u,v}(k)||_{K}\le1$.
    Applying Lemma~\ref{lemma:Bernstein}, the derivative of this function is bounded pointwise by 
    \begin{equation}
        |f'_{u,v}(k)|\le B_N(k)||f_{u,v}(k)||_K\le B_N(k)=\sum_{\ell=1}^{N}P_{y_\ell}(k-x_\ell).
    \end{equation}
    This implies that the operator norm of the derivative matrix satisfies $||R'(k)||\le B_N(k)$. Since $R=2P-I$, we have $R'=2P'$, and
    \begin{equation}
        ||P'||\le\frac{1}{2}B_N(k)
    \end{equation}
    Thus, we get the pointwise upper bound on the quantum metric
    \begin{equation}
        \mathcal{G}(k)\le\frac{1}{2}\mathrm{rank}[(P')^2]||(P')^2||\le||(P')^2||\le\frac14B_N^2(k),
    \end{equation}
    since $\mathrm{rank}[(P')^2]\le2$.
    
    Notice that the Bernstein factor $B_N(k)$ evaluates to a sum of Poisson kernels $P_y(x)={\sinh{y}}/\left({\cosh{y}-\cos{x}}\right)$, which satisfies
    \begin{equation*}
            \frac{1}{2\pi}\int_0^{2\pi}dxP_{y_1}(x-x_1)P_{y_2}(x-x_2)=P_{y_1+y_2}(x_1-x_2).
    \end{equation*}
    Then the upper bound of the quantum metric length is
    \begin{equation}
    \begin{split}
        \int\frac{dk}{2\pi}\mathcal{G}(k)&\le\frac{1}{8\pi}\int dk \sum_{i=1}^N\sum_{j=1}^N P_{y_{i}}(k-x_{i})P_{y_{j}}(k-x_{j})\\
        &=\frac{1}{4}\sum_{i=1}^N\sum_{j=1}^NP_{y_{i}+y_{j}}(x_{i}-x_{j})\\
        &\le\frac{N^2}{4}\coth{y_\mathrm{min}}.
    \end{split}
    \end{equation}
\end{proof}

The equality in Lemma~\ref{lemma:Bernstein} holds if the rational trigonometric polynomial $f$ takes the form of a degree-$N$ Blaschke product
\begin{equation}
    S_N(z)=\prod_{\ell=1}^N\frac{e^{iz}-e^{-y_\ell+ix_\ell}}{1-e^{-y_\ell-ix_\ell}\cdot e^{iz}}.
\end{equation}

We now construct a 2-band, hopping-range-1 Hamiltonian with $N=(n-1)m=1$ that hosts an isolated flat band at $E=0$ and holds the equality in Theorem~\ref{theorem:qml_bound}. Picking $S_1(k)=\left({e^{iz}-e^{-y_0}}\right)/\left({1-e^{-y_0}\cdot e^{iz}}\right)$, we embed the Blaschke product into the 2-band projection operator 
\begin{equation}
    P_f(k)=\frac{1}{2}\left(\begin{array}{cc}
         1&S_1(k)  \\
         \overline{S_1(k)}&1 
    \end{array}\right).
\end{equation}
By construction, $|S_1(k)|^2=1$ on the real axis, so $P_f^2=P_f$ and $\mathrm{tr}(P_f)=1$. Its derivative $P_f'$ is fully ranked and has a pair of eigenvalues $\pm S_1'(k)=\pm B_1(k)$ on the real axis, saturating the pointwise bound from the Bernstein-type inequality.

The range-1 Hamiltonian $H_1(k)$ isolates the flat band at $E_f=0$, with a dispersive band at $E_1(k)>0$. It can be decomposed into
\begin{equation}
    H_1(k)=E_1(k)(I-P_f(k))=\frac{E_1(k)}{2}\left(\begin{array}{cc}
         1&-S_1(k)  \\
         -\overline{S_1(k)}&1 
    \end{array}\right)
\end{equation}
To ensure the hopping range is at most $m=1$, the matrix elements of $H_1(k)$ can only contain the phase factors $e^{ik}$, $1$, and $e^{-ik}$. To achieve this, we must choose $E_1(k)$ to  cancel the denominators of $S_1(k)$ and $\overline{S_1(k)}$. A natural choice is $E_1(k)=C|1-e^{-y_0}|^2=2Ce^{-y_0}(\cosh(y_0)-\cos(k))$, where $C>0$ is an energy scale. The Hamiltonian now reads
\begin{equation}
\begin{split}
    &\quad H_1(k)=C\left(\begin{array}{cc}
         \cos k-\cosh y_0&1-\cos(k-iy_0) \\
         1-\cos(k+iy_0)&\cos k-\cosh y_0
    \end{array}\right).
\end{split}
\end{equation}
This Hamiltonian has the following properties:
\begin{itemize}
    \item Every matrix element in $H_1(k)$ is linear combination of $e^{ik}$, $1$, and $e^{-ik}$.
    \item The eigenvalues of $H_1(k)$ are $E_f=0$ (the flat band) and $E_1(k)=2C(\cosh(y_0)-\cos(k))$ (the dispersive band ). 
    \item The location of the poles saturates the limit allowed by the geometry of the space. The integrated quantum metric evaluates to $\ell_{QM}=\frac{1}{4}\coth y_0$.
\end{itemize}

Following the same protocol, we can further construct a 2-band, range-2 Hamiltonian with $N=(n-1)m=2$ that hosts an isolated flat band at $E=0$ and holds the equality in Theorem~\ref{theorem:qml_bound}. The Hamiltonian now reads
\begin{equation}
\begin{split}
    &\quad H_2(k)=C\left(\begin{array}{cc}
         (\cos k-\cosh y_0)^2&(1-\cos(k-iy_0))^2 \\
         (1-\cos(k+iy_0))^2&(\cos k-\cosh y_0)^2
    \end{array}\right).
\end{split}
\end{equation}
This Hamiltonian has eigenvalues $E_f=0$ (the flat band) and $E_1(k)=2C(\cosh(y_0)-\cos(k))^2$ (the dispersive band ). The integrated quantum metric evaluates to $\ell_{QM}=\coth y_0$.

\end{document}